\documentclass[11pt]{article}

\usepackage{hyperref}
\usepackage{subcaption}
\usepackage[table]{xcolor}
\usepackage{todonotes}
\usepackage{tcolorbox}
\usepackage[margin=1in]{geometry}
\usepackage{amsmath}
\usepackage{enumerate}
\usepackage[shortlabels]{enumitem}
\setlist{noitemsep,topsep=0pt,parsep=0pt,partopsep=0pt}
\usepackage{tikz}
\usepackage{amsfonts}
\usepackage[page]{appendix}
\usepackage{amsthm}
\usepackage[normalem]{ulem}
\newtheorem{theorem}{Theorem}[section]
\usepackage{xfrac}

\newcommand{\one}{({\em i})}
\newcommand{\two}{({\em ii})}
\newcommand{\three}{({\em iii})}

\begin{document}

\title{\bf \Large A Generalised Solution to Distributed Consensus}

\author{
  Heidi Howard, Richard Mortier\\
  University of Cambridge\\
  first.last@cl.cam.ac.uk
}
\date{}

\maketitle

\begin{abstract}
\normalsize
Distributed consensus, the ability to reach agreement in the face of failures and asynchrony, is a fundamental primitive for constructing reliable distributed systems from unreliable components.
The Paxos algorithm is synonymous with distributed consensus, yet it performs poorly in practice and is famously difficult to understand.
In this paper, we re-examine the foundations of distributed consensus.
We derive an abstract solution to consensus, which utilises immutable state for intuitive reasoning about safety.
We prove that our abstract solution generalises over Paxos as well as the Fast Paxos and Flexible Paxos algorithms.
The surprising result of this analysis is a substantial weakening to the quorum requirements of these widely studied algorithms.
\end{abstract}

\section{Introduction}
\label{sec:intro}

We depend upon distributed systems, yet the computers and networks that make up these systems are asynchronous and unreliable.
The longstanding problem of distributed consensus formalises how to reliably reach agreement in such systems.
When solved, we become able to construct strongly consistent distributed systems from unreliable components~\cite{lamport_nets76,schneider_cs90,budhiraja_book93,lampson_wdag96}.
Lamport's Paxos algorithm~\cite{lamport_tcs98} is widely deployed in production to solve distributed consensus~\cite{burrows_osdi06,chandra_podc07}, and experience with it has led to extensive research to improve its performance and our understanding but, despite its popularity, both remain problematic.

Paxos performs poorly in practice because its use of majorities means that each decision requires a round trip to many participants, thus placing substantial load on each participant and the network connecting them.
As a result, systems are typically limited in practice to just three or five participants.
Furthermore, Paxos is usually implemented in the form of \emph{Multi-Paxos}, which establishes one participant as the \emph{master},  introducing a performance bottleneck and increasing latency as all decisions are forwarded via the master.
Given these limitations, many production systems often opt to sacrifice strong consistency guarantees in favour of performance and high availability~\cite{decandia_sosp07,bronson_atc13,lu_sosp15}.
Whilst compromise is inevitable in practical distributed systems~\cite{gilbert_sigact02}, Paxos offers just one point in the space of possible trade-offs.
In response, this paper aims to improve performance by offering a generalised solution allowing  engineers the flexibility to choose their own trade-offs according to the needs of their particular application and deployment environment.

Paxos is also notoriously difficult to understand, leading to much follow up work, explaining the algorithm in simpler terms~\cite{prisco_wdag97,lamport_sigact01,ongaro_atc14,van_cs15} and filling the gaps in the original description, necessary for constructing practical systems~\cite{chandra_podc07,bolosky_nsdi11}.
In recent years, immutability has been increasingly widely utilised in distributed systems to tame complexity~\cite{helland_acmq15}.
Examples such as append-only log stores~\cite{balakrishnan_tcs13,facebook_logdevice} and CRDTs~\cite{shapiro_report11} have inspired us to apply immutability to the problem of consensus.

This paper re-examines the problem of distributed consensus with the aim of improving performance and understanding.
We proceed as follows.
Once we have defined the problem of consensus~(\S\ref{sec:problem}), we propose a generalised solution to consensus that uses only immutable state to enable more intuitive reasoning about correctness~(\S\ref{sec:general}).
We subsequently prove that both Paxos and Fast Paxos~\cite{lamport_msr05} are instances of our generalised consensus algorithm and thus show that both algorithms are conservative in their approach, particularly in their rules for quorum intersection and quorum agreement~(\S\ref{sec:paxos} \&~\S\ref{sec:fast_paxos}).
Finally, we conclude by illustrating the power of our abstraction by outlining three different instances of our generalised consensus algorithm which provide alternative performance trade-offs compared to Paxos~(\S\ref{sec:examples}).

\section{Problem definition}
\label{sec:problem}

The classic formulation of consensus considers how to decide upon a single value in a distributed system.
This seemingly simple problem is made non-trivial by the weak assumptions made about the underlying system: we assume only that the algorithm is correctly executed (i.e.,~the non-Byzantine model).
We do not assume that participants are either reliable or synchronous. Participants may operate at arbitrary speeds and messages may be arbitrarily delayed.

We consider systems comprised of two types of participant: \emph{servers}, which store the value, and \emph{clients}, which read/write the value.
Clients take as input a value to be written and produce as output the value decided by the system.
Messages may only be exchanged between clients and servers and we assume that the set of participants, servers and clients, is fixed and known to the clients.

An algorithm solves consensus if it satisfies the following three requirements:
\begin{itemize}
  \item \textbf{Non-triviality}. All output values must have been the input value of a client.
  \item \textbf{Agreement}. All clients that output a value must output the same value.
  \item \textbf{Progress}. All clients must eventually output a value if the system is reliable and synchronous for a sufficient period.
\end{itemize}

The progress requirement rules out algorithms that could reach deadlock.
As termination cannot be guaranteed in an asynchronous system where failures may occur~\cite{fischer_jacm85}, algorithms need only guarantee termination assuming liveness.



If we have only one server, the solution is straightforward.
The server has a single persistent write-once register, $R0$, to store the decided value.
Clients send requests to the server with their input value.
If $R0$ is unwritten, the value received is written to $R0$ and is returned to the client.
If $R0$ is already written, then the value in $R0$ is read and returned to the client.
The client then outputs the returned value.
This algorithm achieves consensus but requires the server to be available for clients to terminate.
To overcome this limitation requires deployment of more than one server, so we now consider how to generalise to multiple servers.

\section{Generalised solution}
\label{sec:general}

Consider a set of servers, $\{S0,S1,\dots,Sn\}$, where each has a infinite series of write-once, persistent registers, $\{R0,R1,\dots\}$.
Clients read and write registers on servers and, at any time, each register is in one of the three states:
\begin{itemize}
  \item \textbf{unwritten}, the starting state for all registers; or
  \item \textbf{contains a value}, e.g.,~A, B, C; or
  \item \textbf{contains \emph{nil}}, a special value denoted as $\bot$.
\end{itemize}

\begin{figure}
  \centering
  \begin{footnotesize}
    \begin{subfigure}[b]{0.45\textwidth}
      \centering
      \begin{tabular}{ |l|l| }
        \hline
        \textbf{Register} & \textbf{Quorums} \\
        \hline
        $R0$ & $\{\{S0,S1,S2\}\}$ \\
        $R1, R2, \dots $ & $\{\{S0,S1\}, \{S0,S2\}, \{S1,S2\}\}$ \\
        \hline
      \end{tabular}
      \caption{}
      \label{fig:example_configs/3s}
    \end{subfigure}
    \begin{subfigure}[b]{0.45\textwidth}
    \centering
    \begin{tabular}{ |l|l|}
      \hline
      \textbf{Register} & \textbf{Quorums} \\
      \hline
      $R0, R2,\dots$ & $\{\{S0,S1\}\}$ \\
      $R1, R3,\dots$ & $\{\{S2,S3\}\}$\\
      \hline
    \end{tabular}
    \caption{}
    \label{fig:example_configs/4s_disjoint}
    \end{subfigure}
    \begin{subfigure}[b]{0.45\textwidth}
      \centering
      \begin{tabular}{ |l|l|}
        \hline
        \textbf{Register} & \textbf{Quorums} \\
        \hline
        $R0, R1, \dots$ & $\{\{S0,S1\}, \{S2,S3\}\}$ \\
        \hline
      \end{tabular}
      \caption{}
      \label{fig:example_configs/4s_basic}
    \end{subfigure}
    \begin{subfigure}[b]{0.45\textwidth}
    \centering
    \begin{tabular}{ |l|l|}
      \hline
      \textbf{Register} & \textbf{Quorums} \\
      \hline
      $R0, R1,\dots$ & $\{\{S0,S1\}, \{S0,S2\}, \{S1,S2\}\}$ \\
      \hline
    \end{tabular}
    \caption{}
    \label{fig:example_configs/3s_paxos}
    \end{subfigure}
  \caption{Sample configurations for systems of three or four servers.}
  \label{fig:example_configs}
  \end{footnotesize}
\end{figure}

\begin{figure}
  \centering
  \begin{footnotesize}
    \begin{subfigure}[b]{0.3\textwidth}
      \centering
        \begin{tabular}{ |c |c c c|  }
          \hline
             & \textbf{S0} & \textbf{S1} & \textbf{S2} \\
          \hline
          \textbf{R0} & A & $\bot$ & B\\
          \textbf{R1} & $\bot$ & $\bot$ & $\bot$\\
          \textbf{R2} & B & \cellcolor{gray!20}A & \cellcolor{gray!20}A \\
          \hline
        \end{tabular}
      \subcaption{A decided by R2}
      \label{fig:example_state/3s/a}
      \end{subfigure}
      \begin{subfigure}[b]{0.3\textwidth}
        \centering
        \begin{tabular}{ |c |c c c|  }
          \hline
               & \textbf{S0} & \textbf{S1} & \textbf{S2} \\
            \hline
            \textbf{R0} & \cellcolor{gray!20}A & \cellcolor{gray!20}A & \cellcolor{gray!20}A \\
            \textbf{R1} & \cellcolor{gray!20}A & \cellcolor{gray!20}A & \\
            \hline
          \end{tabular}
        \subcaption{A decided by R0 \& R1}
        \label{fig:example_state/3s/b}
        \end{subfigure}
        \begin{subfigure}[b]{0.3\textwidth}
          \centering
          \begin{tabular}{ |c |c c c|  }
            \hline
                 & \textbf{S0} & \textbf{S1} & \textbf{S2} \\
              \hline
              \textbf{R0} & A & $\bot$ & A \\
              \textbf{R1} & A & C & $\bot$\\
              \textbf{R2} &  & C & B\\
              \hline
            \end{tabular}
          \subcaption{No decisions yet}
          \label{fig:example_state/3s/c}
          \end{subfigure}
        %
  \caption{Sample state tables for a system using the configuration in Figure~\ref{fig:example_configs/3s}.}
  \label{fig:example_state}
    \end{footnotesize}
\end{figure}

A quorum, $Q$, is a (non-empty) subset of servers, such that if all servers have a same (non-nil) value $v$ in the same register then $v$ is said to be \emph{decided}.
A \emph{register set}, $i$, is the set comprised of the register $Ri$ from each server.
Each register set $i$ is configured with a set of quorums, $\mathcal{Q}_i$, and four example configurations are given in Figure~\ref{fig:example_configs}.
The state of all registers can be represented in a table, known as a \emph{state table}, where each column represents the state of one server and each row represents a register set.
By combining a configuration with a state table, we can determine whether any decision(s) have been reached, as shown in Figure~\ref{fig:example_state}.

\subsection{Correctness}
\label{subsec:general/correctness}

\begin{figure}
  \begin{tcolorbox}
  \textbf{Rule 1: Quorum agreement}. A client may only output a (non-nil) value $v$ if it has read $v$ from a quorum of servers in the same register set.

  \textbf{Rule 2: New value}. A client may only write a (non-nil) value $v$ provided that either $v$ is the client's input value or that the client has read $v$ from a register.

  \textbf{Rule 3: Current decision}. A client may only write a (non-nil) value $v$ to register $r$ on server $s$ provided that if $v$ is decided in register set $r$ by a quorum $Q \in \mathcal{Q}_r$ where $s\in Q$ then no value $v'$ where $v \neq v'$ can also be decided in register set $r$.

  \textbf{Rule 4: Previous decisions}. A client may only write a (non-nil) value $v$ to register $r$ provided no value $v'$ where $v \neq v'$ can be decided by the quorums in register sets $0$ to $r-1$.

  \end{tcolorbox}
  \caption{The four rules for correctness.}
  \label{fig:rules}
\end{figure}

Figure~\ref{fig:rules} describes a generalised solution to consensus by giving four rules governing how clients interact with registers to ensure that the non-triviality and agreement requirements for consensus~(\S\ref{sec:problem}) are satisfied.

Rule 1 (\emph{quorum agreement}) ensures that clients only output values that have been decided.
Rule 2 (\emph{new value}) ensures that only client input values can be written to registers thus only client input values can be decided and output by clients.
Rules 3 and 4 ensures that no two quorums can decide upon different values.
Rule 3 (\emph{current decision}) ensures that all decisions made by a register set will be for the same value whilst Rule 4 (\emph{previous decisions}) ensures that all decisions made by different register sets are for the same value.

\subsection{Implementing the correctness rules}
\label{sec:general/safety}

Rules 1 and 2 are easy to implement, but Rules 3 and 4 require more careful treatment.

\begin{figure}
  \centering
  \begin{footnotesize}
    \begin{tabular}{ |l|l| }
      \hline
      \textbf{Register} & \textbf{Client} \\
      \hline
      $R0, R3,\dots$ & C0 \\
      $R1, R4,\dots$ & C1 \\
      $R2, R5,\dots$ & C2 \\
      \hline
    \end{tabular}
  \caption{Sample round robin allocation of register sets to clients.}
  \label{fig:client_configs}
  \end{footnotesize}
\end{figure}

\paragraph{Rule 3 (\emph{current decision}).}

The simplest implementation of Rule 3 is to permit only configurations with one quorum per register set, as in Figure~\ref{fig:example_configs/4s_disjoint}.
We generalise this to \emph{intersecting quorums configurations} by permitting multiple quorums per register set, provided that all quorums for a given register set intersect, as in Figure~\ref{fig:example_configs/3s_paxos}.
The requirement for intersection ensures that if multiple quorums in a register set decide a value then they must decide the same value as they must share a common register.

Alternatively, we can support disjoint quorums if we require that all values written to a given register set are the same.
This can be achieved by assigning register sets to clients and requiring that clients write only to their own register sets, with at most one value.
In practice, this could be implemented by using an allocation such as that in Figure~\ref{fig:client_configs} and by requiring clients to keep a persistent record of which register sets they have written too.
We refer to these as \emph{client restricted configurations}.

Both techniques, intersecting quorums configurations and client restricted configurations, can be combined on a per-register-set basis.

\paragraph{Rule 4 (\emph{previous decisions}).}

Rule 4 requires clients to ensure that, before writing a (non-nil) value, previous register sets cannot decide a different value.
This is trivially satisfied for register set $0$, however, more work is required by clients to satisfy this rule for subsequent register sets.

Assume each client maintains their own local copy of the state table.
Initially, each client's state table is empty as they have not yet learned anything regarding the state of the servers.
A client can populate its state tables by reading registers and storing the results in its copy of the state table.
Since the registers are persistent and write-once, if a register contains a value (nil or otherwise) then any reads will always remain valid.
Each client's state tables will therefore always contain a subset of the values from the state table.

From its local state table, each client can track whether decisions have been reached or could be reached by previous quorums.
We refer to this as the \emph{decision table}. At any given time, each quorum is in one of four decision states:
\begin{itemize}
  \item \textsc{Any}: Any value could be decided by this quorum.
  \item \textsc{Maybe~$v$}: If this quorum reaches a decision, then value $v$ will be decided.
  \item \textsc{Decided~$v$}: The value $v$ has been decided by this quorum; a final state.
  \item \textsc{None}: This quorum will not decide a value; a final state.
\end{itemize}

The rules for updating the decision table are as follows:
Initially, the decision state of all quorums is \textsc{Any}.
If there is a quorum where all registers contain the same value $v$ then its decision state is \textsc{Decided}~$v$.
When a client reads \emph{nil} from register $r$ on server $s$ then for all quorums $Q \in \mathcal{Q}_r$ where $s \in Q$, the decision state \textsc{Any}/\textsc{Maybe}~$v$ becomes \textsc{None}.
When a client reads a non-\emph{nil} value $v$ from a client restricted register set $r$ then for all quorums over register sets $0$ to $r$, the decision state \textsc{Any} becomes \textsc{Maybe}~$v$ and \textsc{Maybe}~$v'$ where $v\neq v'$ becomes \textsc{None}.
When a client reads a non-\emph{nil} value $v$ from a quorum intersecting register set $r$ on server $s$ then for all quorums $Q \in \mathcal{Q}_r$ where $s \in Q$ and for all quorums over register sets $0$ to $r-1$, the state \textsc{Any} becomes \textsc{Maybe}~$v$ and \textsc{Maybe}~$v'$ where $v\neq v'$ becomes \textsc{None}.

These rules utilise the knowledge that if a client reads a (non-nil) value $v$ from the register $r$ on server $s$, it learns that:
\begin{itemize}
  \item If $r$ is client restricted then all quorums in $r$ must decide $v$ if they reach a decision (Rule 3).
  \item If any quorum of register sets $0$ to $r-1$ reaches a decision then value $v$ is decided (Rule 4).
\end{itemize}

\begin{figure}
  \begin{tcolorbox}
    A client may output value $v$ provided at least one quorum state is \textsc{Decided}~$v$ (Rule 1).

    A client $c$ may write a non-nil value $v$ to register set $r$ provided:
    \begin{enumerate}[i.]
      \item $v$ is $c$'s input value or has been read from a register (Rule 2), and
      \item $r$ is either:
      \begin{itemize}
        \item quorum intersecting, or
        \item client restricted and $r$ has been allocated to $c$ but not yet used (Rule 3), and
      \end{itemize}
      \item the decision state of each quorum from register sets $0$ to $r-1$ is \textsc{None}, \textsc{Maybe}~$v$ or \textsc{Decided}~$v$ (Rule 4).
    \end{enumerate}
  \end{tcolorbox}
  \caption{Client decision table rules}
  \label{fig:decisiontable}
\end{figure}

Figure~\ref{fig:decisiontable} describes how clients can use decision tables to implement the four rules for correctness. 

\subsection{Examples}
\label{subsec:general/examples}

\begin{figure}
  \centering
  \begin{footnotesize}
  \begin{subfigure}[b]{\textwidth}
    \centering
    \begin{tabular}{ |c |c c c c | }
      \hline
         & \textbf{S0} & \textbf{S1} & \textbf{S2} & \textbf{S3} \\
      \hline
      \textbf{R0} &  &  &  &  \\
      \hline
    \end{tabular}
    \hspace{1cm}
    \begin{tabular}{| l l l |}
      \hline
      \textbf{Register} & \textbf{Quorum} & \textbf{Decision state} \\
      \hline
      $R0$ & $\{S0,S1\}$ & Any\\
      \hline
    \end{tabular}
    \caption{Initial state.}
    \label{fig:oneq_decision/initial}
  \end{subfigure}

  \begin{subfigure}[b]{\textwidth}
    \centering
    \begin{tabular}{ |c |c c c c | }
      \hline
         & \textbf{S0} & \textbf{S1} & \textbf{S2} & \textbf{S3} \\
      \hline
      \textbf{R0} &  &  &  &  \\
      \textbf{R1} &  &  &  & B \\
      \hline
    \end{tabular}
    \hspace{1cm}
    \begin{tabular}{| l l l |}
      \hline
      \textbf{Register} & \textbf{Quorum} & \textbf{Decision state} \\
      \hline
      $R0$ & $\{S0,S1\}$ & \textsc{Maybe B}\\
      $R1$ & $\{S2,S3\}$ & \textsc{Maybe B}\\
      \hline
    \end{tabular}
    \caption{State after reading B from $R1$ on $S3$.}
    \label{fig:oneq_decision/readone}
  \end{subfigure}

  \begin{subfigure}[b]{\textwidth}
    \centering
    \begin{tabular}{ |c |c c c c | }
      \hline
         & \textbf{S0} & \textbf{S1} & \textbf{S2} & \textbf{S3} \\
      \hline
      \textbf{R0} & A &  &  &  \\
      \textbf{R1} &  &  &  & B \\
      \hline
    \end{tabular}
    \hspace{1cm}
    \begin{tabular}{| l l l |}
      \hline
      \textbf{Register} & \textbf{Quorum} & \textbf{Decision state} \\
      \hline
      $R0$ & $\{S0,S1\}$ & \textsc{None}\\
      $R1$ & $\{S2,S3\}$ & \textsc{Maybe B}\\
      \hline
    \end{tabular}
    \caption{State after reading A from $R0$ on $S0$.}
    \label{fig:oneq_decision/readtwo}
  \end{subfigure}

  \begin{subfigure}[b]{\textwidth}
    \centering
    \begin{tabular}{ |c |c c c c | }
      \hline
         & \textbf{S0} & \textbf{S1} & \textbf{S2} & \textbf{S3} \\
      \hline
      \textbf{R0} & A &  &  &  \\
      \textbf{R1} &  &  & B &  B\\
      \hline
    \end{tabular}
    \hspace{1cm}
    \begin{tabular}{| l l l |}
      \hline
      \textbf{Register} & \textbf{Quorum} & \textbf{Decision state} \\
      \hline
      $R0$ & $\{S0,S1\}$ & \textsc{None}\\
      $R1$ & $\{S2,S3\}$ & \textsc{Decided B}\\
      \hline
    \end{tabular}
    \caption{State after reading read B from $R1$ on $S2$.}
    \label{fig:oneq_decision/decided}
  \end{subfigure}
  \caption{Sample client state tables (left) and decision tables (right).}
  \label{fig:oneq_decision}
  \end{footnotesize}
\end{figure}

This process is illustrated by Figures~\ref{fig:oneq_decision} and ~\ref{fig:disjoint_decision}, which demonstrate how a client's state is updated as they read registers.
Figure~\ref{fig:oneq_decision} shows the state of a client $C0$ in a system of 4 servers using the intersecting quorum configuration from Figure~\ref{fig:example_configs/4s_disjoint}.
Figure~\ref{fig:oneq_decision/initial} shows the client's initial state.
The client's state table is empty thus the status of all quorums in the decision table is \textsc{Any}.
At this time, the client may only write non-nil values to $R0$ due to condition (iii) in Figure~\ref{fig:decisiontable}.
Next, Figure~\ref{fig:oneq_decision/readone}, the status of quorum $\{S2,S3\}$ over register set $1$ is updated to \textsc{Maybe}~B since, depending on the state of register $R1$ on $S2$, either this quorum will not reach a decision or it decides value B.
Since the client that wrote B into $R1$ on $S3$ must have followed Rule 4, the quorum in $R0$ must decide B if it reaches a decision thus its status is updated to \textsc{Maybe}~B.
The client $C0$ can now write value B to $R1$ or $R2$.
Subsequently in Figure~\ref{fig:oneq_decision/readtwo}, the client could now safely write its input value to $R1$ but there would be no use in doing so.
Finally in Figure~\ref{fig:oneq_decision/decided}, the client learns that B is decided and thus can output B.

\begin{figure}
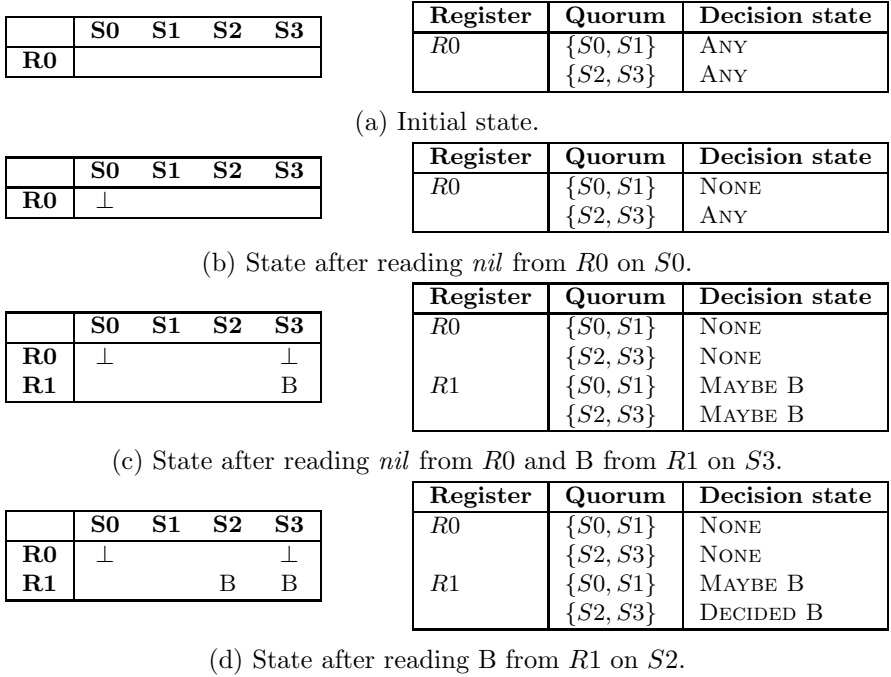

  \centering
  \begin{footnotesize}
  \begin{subfigure}[b]{\textwidth}
    \centering
    \begin{tabular}{ |c|cccc| }
      \hline
         & \textbf{S0} & \textbf{S1} & \textbf{S2} & \textbf{S3} \\
      \hline
      \multicolumn{1}{|c|}{\textbf{R0}} &  &  &  &  \\
      \hline
    \end{tabular}
    \hspace{1cm}
    \begin{tabular}{ |l |l | l |}
      \hline
      \textbf{Register} & \textbf{Quorum} & \textbf{Decision state} \\
      \hline
      $R0$ & $\{S0,S1\}$ & \textsc{Any} \\
      & $\{S2,S3\}$ & \textsc{Any} \\
      \hline
    \end{tabular}
    \caption{Initial state.}
    \label{fig:disjoint_decision/initial}
  \end{subfigure}

  \begin{subfigure}[b]{\textwidth}
    \centering
    \begin{tabular}{ |c|cccc| }
      \hline
         & \textbf{S0} & \textbf{S1} & \textbf{S2} & \textbf{S3} \\
      \hline
      \multicolumn{1}{|c|}{\textbf{R0}} & $\bot$ &  &  &  \\
      \hline
    \end{tabular}
    \hspace{1cm}
    \begin{tabular}{ |l |l | l |}
      \hline
      \textbf{Register} & \textbf{Quorum} & \textbf{Decision state} \\
      \hline
      $R0$ & $\{S0,S1\}$ & \textsc{None} \\
      & $\{S2,S3\}$ & \textsc{Any} \\
      \hline
    \end{tabular}
    \caption{State after reading \emph{nil} from $R0$ on $S0$.}
    \label{fig:disjoint_decision/one}
  \end{subfigure}

  \begin{subfigure}[b]{\textwidth}
    \centering
    \begin{tabular}{ |c|cccc| }
      \hline
         & \textbf{S0} & \textbf{S1} & \textbf{S2} & \textbf{S3} \\
      \hline
      \textbf{R0} & $\bot$ &  &  & $\bot$ \\
      \textbf{R1} &  &  &  & B \\
      \hline
    \end{tabular}
    \hspace{1cm}
    \begin{tabular}{ |l |l | l |}
      \hline
      \textbf{Register} & \textbf{Quorum} & \textbf{Decision state} \\
      \hline
      $R0$ & $\{S0,S1\}$ & \textsc{None} \\
      & $\{S2,S3\}$ & \textsc{None} \\
      $R1$ & $\{S0,S1\}$ & \textsc{Maybe B} \\
      & $\{S2,S3\}$ & \textsc{Maybe B} \\
      \hline
    \end{tabular}
    \caption{State after reading \emph{nil} from $R0$ and B from $R1$ on $S3$.}
    \label{fig:disjoint_decision/three}
  \end{subfigure}

  \begin{subfigure}[b]{\textwidth}
    \centering
    \begin{tabular}{ |c|cccc| }
      \hline
         & \textbf{S0} & \textbf{S1} & \textbf{S2} & \textbf{S3} \\
      \hline
      \textbf{R0} & $\bot$ &  &  & $\bot$ \\
      \textbf{R1} &  &  & B & B \\
      \hline
    \end{tabular}
    \hspace{1cm}
    \begin{tabular}{ |l |l | l |}
      \hline
      \textbf{Register} & \textbf{Quorum} & \textbf{Decision state} \\
      \hline
      $R0$ & $\{S0,S1\}$ & \textsc{None} \\
      & $\{S2,S3\}$ & \textsc{None} \\
      $R1$ & $\{S0,S1\}$ & \textsc{Maybe B} \\
      & $\{S2,S3\}$ & \textsc{Decided B} \\
      \hline
    \end{tabular}
    \caption{State after reading B from $R1$ on $S2$.}
    \label{fig:disjoint_decision/four}
  \end{subfigure}
  \caption{Sample client state tables (left) and decision tables (right).}
  \label{fig:disjoint_decision}
  \end{footnotesize}
\end{figure}

Figure~\ref{fig:disjoint_decision} shows the state of a client $C0$ in a system using the configuration in Figure~\ref{fig:example_configs/4s_basic} and the client allocation from Figure~\ref{fig:client_configs}.
Figure~\ref{fig:disjoint_decision/initial} shows the initial state of the client $C0$.
At this time, the client $C0$ can only write non-nil values to $R0$.
Later in Figure~\ref{fig:disjoint_decision/three}, the client has updated the status of both quorums in $R1$ to \textsc{Maybe}~B after reading B from $R1$. This is because register set $1$ is client restricted to value B.

\section{Generalisation of Paxos}
\label{sec:paxos}

\begin{figure}
\begin{tcolorbox}
\textbf{Phase 1}
\begin{itemize}[leftmargin=*]
  \item A client $c$ chooses a register set $r$ that it has been assigned but not yet used and sends \textsc{P1a}($r$) to all servers.
  \item Upon receiving \textsc{p1a}($r$), each server checks if register $r$ is unwritten.
  If so, any unwritten registers up to $r-1$ (inclusive) are set to \emph{nil}.
  The server replies with \textsc{p1b}($r$, $S$) where $S$ is a set of all written non-nil registers.
  \item If $c$ receives \textsc{p1b} messages from a majority of servers then $c$ chooses the value from the greatest (non-nil) register.
  If no values were returned with \textsc{P1b} messages then $c$ chooses its input value.
  $c$ then proceeds to phase two.
  Otherwise, $c$ times out and restarts phase one.
\end{itemize}

\textbf{Phase 2}
\begin{itemize}[leftmargin=*]
  \item $c$ sends \textsc{P2a}($r$, $v$) to all servers where $v$ is value chosen at the end of phase one.
  \item Upon receiving \textsc{P2a}($r$, $v$), each server checks if register $r$ is unwritten.
  If so, any unwritten registers up to $r-1$ (inclusive) are set to to \emph{nil} and register $r$ is set to $v$.
  The server replies with \textsc{P2b}($r$, $v$).
  \item If $c$ receives \textsc{P2b} messages from the majority of servers then $c$ learns that the value $v$ has been decided and can output $v$.
  Otherwise, $c$ times out and restarts phase one.
\end{itemize}
\end{tcolorbox}
\caption{The Paxos algorithm using write-once registers.}
\label{fig:paxos_alg}
\end{figure}

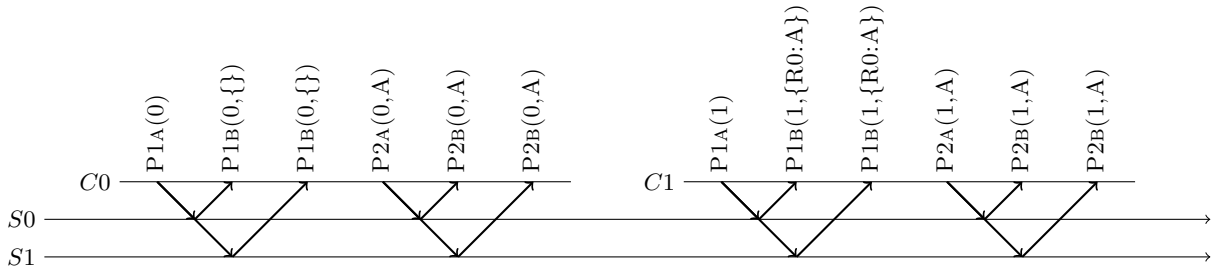
\begin{figure}
  \centering
  \begin{footnotesize}
  \begin{tikzpicture}
    \draw [-] (0.5,1) node[left] {$C0$} -- (6.5,1);
    \draw [-] (8,1) node[left] {$C1$} -- (14,1);
    \draw [->] (-0.5,0.5) node[left] {$S0$} -- (15,0.5);
    \draw [->] (-0.5,0) node[left] {$S1$} -- (15,0);

    \draw [->,thick] (1,1) node[rotate=90,above,anchor=west] {\textsc{P1a(0)}} -- (1.5,0.5);
    \draw [<-,thick] (2,1) node[rotate=90,above,anchor=west] {\textsc{P1b(0,\{\})}} -- (1.5,0.5);
    \draw [->,thick] (1,1) -- (2,0);
    \draw [<-,thick] (3,1) node[rotate=90,above,anchor=west] {\textsc{P1b(0,\{\})}} -- (2,0);

    \draw [->,thick] (4,1) node[rotate=90,above,anchor=west] {\textsc{P2a(0,A)}} -- (4.5,0.5);
    \draw [<-,thick] (5,1) node[rotate=90,above,anchor=west] {\textsc{P2b(0,A)}} -- (4.5,0.5);
    \draw [->,thick] (4,1) -- (5,0);
    \draw [<-,thick] (6,1) node[rotate=90,above,anchor=west] {\textsc{P2b(0,A)}} -- (5,0);

    \draw [->,thick] (1+7.5,1) node[rotate=90,above,anchor=west] {\textsc{P1a(1)}} -- (1.5+7.5,0.5);
    \draw [<-,thick] (2+7.5,1) node[rotate=90,above,anchor=west] {\textsc{P1b(1,\{R0:A\})}} -- (1.5+7.5,0.5);
    \draw [->,thick] (1+7.5,1) -- (2+7.5,0);
    \draw [<-,thick] (3+7.5,1) node[rotate=90,above,anchor=west] {\textsc{P1b(1,\{R0:A\})}} -- (2+7.5,0);

    \draw [->,thick] (4+7.5,1) node[rotate=90,above,anchor=west] {\textsc{P2a(1,A)}} -- (4.5+7.5,0.5);
    \draw [<-,thick] (5+7.5,1) node[rotate=90,above,anchor=west] {\textsc{P2b(1,A)}} -- (4.5+7.5,0.5);
    \draw [->,thick] (4+7.5,1) -- (5+7.5,0);
    \draw [<-,thick] (6+7.5,1) node[rotate=90,above,anchor=west] {\textsc{P2b(1,A)}} -- (5+7.5,0);
  \end{tikzpicture}
\end{footnotesize}
  \caption{Sample message exchange for Paxos}
  \label{fig:paxos_msd}
\end{figure}

The (unoptimised) Paxos algorithm is described in Figure~\ref{fig:paxos_alg} using only write-once registers.
Figure~\ref{fig:paxos_msd} gives an example of the message exchange as two clients execute Paxos with three servers.

\begin{figure}
  \centering
  \begin{footnotesize}
  \begin{subfigure}[b]{\textwidth}
    \centering
    \begin{tabular}{|c|ccc|}
      \hline
         & \textbf{S0} & \textbf{S1} & \textbf{S2} \\
      \hline
      \multicolumn{1}{|c|}{\textbf{R0}} &  &  &   \\
      \hline
    \end{tabular}
    \hspace{1cm}
    \begin{tabular}{|lll|}
      \hline
      \textbf{Register} & \textbf{Quorum} & \textbf{Decision state} \\
      \hline
      $R0$ & $\{S0,S1\}$ & \textsc{Any} \\
      & $\{S0,S2\}$ & \textsc{Any} \\
      & $\{S1,S2\}$ & \textsc{Any} \\
      \hline
    \end{tabular}
    \caption{Initial state, unchanged after receiving \textsc{P1b}(0,$\{\}$) from $S1$.}
    \label{fig:paxos_decision/initial}
  \end{subfigure}

  \begin{subfigure}[b]{\textwidth}
    \centering
    \begin{tabular}{|c|ccc|}
      \hline
         & \textbf{S0} & \textbf{S1} & \textbf{S2} \\
      \hline
      \multicolumn{1}{|c|}{\textbf{R0}} & A & A &  \\
      \hline
    \end{tabular}
    \hspace{1cm}
    \begin{tabular}{|lll|}
      \hline
      \textbf{Register} & \textbf{Quorum} & \textbf{Decision state} \\
      \hline
      $R0$ & $\{S0,S1\}$ & \textsc{Decided A} \\
      & $\{S0,S2\}$ & \textsc{Maybe A} \\
      & $\{S1,S2\}$ & \textsc{Maybe A} \\
      \hline
    \end{tabular}
    \caption{State after receiving \textsc{P2b}(0,A) from $S1$.}
    \label{fig:paxos_decision/end}
  \end{subfigure}
  \caption{Sample client state tables (left) and decision tables (right) for client $C0$ during the execution in Figure~\ref{fig:paxos_msd}.}
  \label{fig:paxos_decision}
  \end{footnotesize}
\end{figure}

\begin{figure}
  \centering
  \begin{footnotesize}

  \begin{subfigure}[b]{\textwidth}
    \centering
    \begin{tabular}{| c |c c c | }
      \hline
         & \textbf{S0} & \textbf{S1} & \textbf{S2} \\
      \hline
      \multicolumn{1}{|c|}{\textbf{R0}} & A & &  \\
      \hline
    \end{tabular}
    \hspace{1cm}
    \begin{tabular}{ |l l l|}
      \hline
      \textbf{Register} & \textbf{Quorum} & \textbf{Decision state} \\
      \hline
      $R0$ & $\{S0,S1\}$ & \textsc{Maybe A} \\
      & $\{S0,S2\}$ & \textsc{Maybe A} \\
      & $\{S1,S2\}$ & \textsc{Maybe A} \\
      \hline
    \end{tabular}
    \caption{State after receiving \textsc{P1b(1,\{R0:A\})} from $S0$.}
    \label{fig:paxostwo_decision/midone}
  \end{subfigure}

  \begin{subfigure}[b]{\textwidth}
    \centering
    \begin{tabular}{ |c |c c c | }
      \hline
         & \textbf{S0} & \textbf{S1} & \textbf{S2} \\
      \hline
      \multicolumn{1}{|c|}{\textbf{R0}} & A & A &  \\
      \hline
    \end{tabular}
    \hspace{1cm}
    \begin{tabular}{ |l l  l |}
      \hline
      \textbf{Register} & \textbf{Quorum} & \textbf{Decision state} \\
      \hline
      $R0$ & $\{S0,S1\}$ & \textsc{Decided A} \\
      & $\{S0,S2\}$ & \textsc{Maybe A} \\
      & $\{S1,S2\}$ & \textsc{Maybe A} \\
      \hline
    \end{tabular}
    \caption{State after receiving \textsc{P1b(1,\{R0:A\})} from $S1$.}
    \label{fig:paxostwo_decision/endone}
  \end{subfigure}

  \caption{Sample client state tables (left) and decision tables (right) for client $C1$ during the execution in Figure~\ref{fig:paxos_msd}.}
  \label{fig:paxostwo_decision}
    \end{footnotesize}
\end{figure}

We observe that Paxos is a conservative instance of our generalised solution to consensus.
The configuration used by Paxos is majorities for all register sets, such a configuration is given in Figure~\ref{fig:example_configs/3s_paxos}.
Paxos also uses client restricted for all register sets and a suitable client assignment is given in Figure~\ref{fig:client_configs}.
The purpose of phase one is to implement Rule 4 and the purpose of phase two is to implement Rule 1.
Earlier (\S\ref{sec:general}), we proposed client state and decision tables as a mechanism for clients to implement the rules for correctness.
Upon receiving \textsc{P1b}($r$,$\mathcal{R}$) where $\mathcal{R}$ is the set of registers from a server, the client learns the contents of registers $0$ to $r-1$.
This is because registers are always written to in-order on each server and register $r$ must be unwritten.
Therefore the client's state table and thus its decision table can be updated accordingly.
This is demonstrated in Figure~\ref{fig:paxos_decision} for client $C0$ and Figure~\ref{fig:paxostwo_decision} for client $C1$.

\subsection{Weakened quorum intersection requirements}
\label{subsec:paxos/quorums}

The boolean function \emph{I} tests whether two or more quorum sets are intersecting and is defined as $I(\{\mathcal{Q}^i\}) \equiv \forall i, \forall Q^i \in \mathcal{Q}^i: \bigcap Q^i \neq \emptyset$.

Paxos utilises majorities as it requires all quorums, $Q \in \mathcal{Q}$, to intersect, regardless of the register set or phase of the algorithm.
That is, in terms of \emph{I}, $I(\mathcal{Q},\mathcal{Q})$.

Instead, we differentiate between the quorums used for each register set and which phase of Paxos the quorum is used for.
$\mathcal{Q}_r^k$ is the set of quorums for phase $k$ of register set $r$.
We observe that quorum intersection is required only between the phase one quorum for register set $r$ and the phase two quorums of register sets 0 to $r-1$.
This is the case because a client can always proceed to phase two after intersecting with all previous phase two quorums since the condition (iii) in Figure~\ref{fig:decisiontable} will be satisfied.
More formally, \par
\hfill $\forall r \in \mathbb{N}_{0}, \forall r' \in \mathbb{N}_{<r}: I(\mathcal{Q}_r^1,\mathcal{Q}_{r'}^2)$. \hfill (*)

This result confirms the findings of Flexible Paxos~\cite{howard_opodis16}.
This is illustrated in Figure~\ref{fig:paxos_decision/initial} where the client was safe to proceed to phase two from startup since there is no intersection requirement for register set $0$.

\subsection{Progress without quorums}
\label{subsubsec:paxos/skip}

Each of the two phases of Paxos waits for agreement from a quorum of servers.
However, we observe that it may be possible to proceed prior to reaching quorum agreement.

A client can safely terminate once it learns that a value has been decided (Rule 1).
This need not be the result of completing both phases of the algorithm.
This is illustrated in Figure~\ref{fig:paxostwo_decision/endone} where the client learns that value A has been decided prior to starting phase two.

Similarly, if a server learns that a register $r$ contains a (non-nil) value $v$ then it also learns that if any quorums from register sets $0$ to $r$ reach a decision then $v$ must be chosen.
By updating their decision table, we observe that it is no longer necessary for the client in phase one to intersect with the phase two quorums of registers up to $r$ (inclusive).
This is illustrated in Figure~\ref{fig:paxostwo_decision/midone} where the client could safely proceed to phase two after one \textsc{P1b} message as the client reads a non-nil value from predecessor register set.

\section{Generalisation of Fast Paxos}
\label{sec:fast_paxos}

\begin{figure}
  \centering
  \begin{footnotesize}
    \begin{subfigure}[b]{0.45\textwidth}
      \centering
      \begin{tabular}{ |l|l|}
        \hline
        \textbf{Register} & \textbf{Quorums} \\
        \hline
        $R0, R1, \dots$ & $\{\{S0,S1,S2\}, \{S0,S1,S3\},$ \\
        & $\{S0,S2,S3\}, \{S1,S2,S3\} \}$ \\
         \hline
      \end{tabular}
      \caption{}
      \label{fig:example_configs/4s_intersecting}
    \end{subfigure}
    \begin{subfigure}[b]{0.45\textwidth}
      \centering
      \begin{tabular}{ |l|l| }
        \hline
        \textbf{Register} & \textbf{Quorums} \\
        \hline
        $R0, R1, R2$ & $\{\{S0,S1,S2\}\}$ \\
        $R3, R4, \dots $ & $\{\{S0,S1\}, \{S0,S2\}, \{S1,S2\}\}$ \\
        \hline
      \end{tabular}
      \caption{}
      \label{fig:add_example_configs/allin}
    \end{subfigure}
    \begin{subfigure}[b]{0.45\textwidth}
    \centering
    \begin{tabular}{ |l|l|}
      \hline
      \textbf{Register} & \textbf{Quorums} \\
      \hline
      $R0$ & $\{\{S0,S1\}\}$ \\
      $R1, R2, \dots $ & $\{\{S0,S1\}, \{S0,S2\}, \{S1,S2\}\}$ \\
      \hline
    \end{tabular}
    \caption{}
    \label{fig:add_example_configs/fixedmaj}
    \end{subfigure}
    \begin{subfigure}[b]{0.45\textwidth}
      \centering
      \begin{tabular}{ |l|l|}
        \hline
        \textbf{Register} & \textbf{Quorums} \\
        \hline
        $R0, R1, \dots, R10$ & $\{\{S0,S1\}, \{S0,S2\}, \{S1,S2\}\}$  \\
        $R11, \dots$ & $\{\{S3,S4\}, \{S3,S5\}, \{S4,S5\}\}$  \\
        \hline
      \end{tabular}
      \caption{}
      \label{fig:add_example_configs/reconfig}
    \end{subfigure}
  \caption{Additional sample configurations.}
  \label{fig:add_example_configs}
  \end{footnotesize}
\end{figure}

Paxos requires client restricted configuration for all register sets.
Fast Paxos~\cite{lamport_msr05} generalises Paxos by permitting intersecting quorum configurations for some register sets, known as \emph{fast} sets, whilst still utilising client restricted configurations for remaining sets, known as \emph{classic} sets.
Quorums for \emph{classic} sets must include $>\sfrac{1}{2}$ of servers whereas quorums for \emph{fast} sets must include $\geq\sfrac{3}{4}$ of servers.
Figure~\ref{fig:example_configs/4s_intersecting} is an example of such a configuration.

Fast Paxos modifies our original Paxos algorithm (Figure~\ref{fig:paxos_alg}) as follows:
\begin{itemize}
  \item If a register set is \emph{fast} then a client does not need to be assigned the register set nor do they need to ensure that they write to it with only one value.
  Any client can use the any \emph{fast} register set.
  \item If the register set is \emph{fast} then completion of each phase requires responses from $\sfrac{3}{4}$ of servers instead of $\sfrac{1}{2}$ of servers.
  \item  When choosing a value at the end of phase one, multiple values may have been read from the same register set (if it was a \emph{fast} set), in which case the client chooses the most common.
\end{itemize}

\subsection{Weakened quorum intersection requirements}
\label{subsec:fast_paxos/quorums}

Fast Paxos uses quorums of $\sfrac{3}{4}$ of servers for \emph{fast} sets and $\sfrac{1}{2}$ of servers for \emph{classic} sets since it requires the following intersection between quorums for \emph{fast} sets, $\mathcal{Q}_f$ and quorums for \emph{classic} sets, $\mathcal{Q}_c$: $I(\mathcal{Q}_c,\mathcal{Q}_c)$, and $I(\mathcal{Q}_c,\mathcal{Q}_f,\mathcal{Q}_f)$.\footnote{Generalisation to quorums requires us to rewrite the value selection rule to chose the value which may be decided. }

As with Paxos, these intersection requirements are conservative.
We differentiate between the quorums used for each register set and which phase of the algorithm the quorum is used for.
$\mathcal{Q}_r^k$ is the set of quorums for phase $k$ of register set $r$.
In addition to Paxos's weakened intersection requirement (Eq.~(*)), we observe that two additional quorum intersections are required: between the quorums for each \emph{fast} register set, and between the phase one quorum for register set $r$ and any pair of phase two quorums of \emph{fast} register sets from 0 to $r-1$.
Denoting the set of \emph{fast} register sets as $\mathbb{F}$, we express these requirements as follows: \par
\hfill $\forall r \in \mathbb{F}: I(\mathcal{Q}_r^2,\mathcal{Q}_r^2)$ and $ \forall r \in \mathbb{N}_{0}, \forall r' \in \mathbb{F}_{<r}: I(\mathcal{Q}_r^1,\mathcal{Q}_{r'}^2,\mathcal{Q}_{r'}^2)$. \hfill (**)

\subsection{Progress without quorums}
\label{subsec:fast_paxos/skip}

Utilising decision tables, we observe that quorum agreement is sufficient but not necessary for a client to complete a phase of the algorithm.
In particular, during the following three cases.

\one~As with Paxos, once a client learns that a quorum of registers contain a value then the client can terminate and return that value.

\two~If a client learns that a register $r$ contains a (non-nil) value $v$ then it also learns that if any quorums from register sets 0 to $r-1$ reach a decision then $v$ must be chosen.
If $r$ is a \emph{classic} register set then it also learns that if any quorums from register sets $r$ reach a decision then $v$ must be chosen.
The client therefore no longer needs to intersect with quorums $0$ to $r-1$ if $r$ is \emph{fast} or quorums $0$ to $r$ if $r$ is \emph{classic}.

\begin{figure}
  \centering
  \begin{footnotesize}

    \begin{tabular}{|c |c c c c|}
      \hline
         & \textbf{S0} & \textbf{S1} & \textbf{S2} & \textbf{S3} \\
         \hline
      \textbf{R0} & $\bot$ & $\bot$ &  & \\
      \hline
    \end{tabular}
    \hspace{1cm}
    \begin{tabular}{ |l l  l |}
      \hline
      \textbf{Register} & \textbf{Quorum} & \textbf{Decision state} \\
      \hline
      $R0$ & $\{S0,S1,S2\}$ & \textsc{None} \\
      & $\{S0,S1,S3\}$ & \textsc{None} \\
      & $\{S0,S2,S3\}$ & \textsc{None} \\
      & $\{S1,S2,S3\}$ & \textsc{None} \\
      \hline
    \end{tabular}

  \caption{Sample client state table (left) and decision table (right) for client $C0$ during Fast Paxos.}
  \label{fig:fpaxos_decision}
  \end{footnotesize}
\end{figure}

\begin{figure}
  \centering
  \begin{footnotesize}

    \begin{tabular}{| c | c c c  c|}
      \hline
         & \textbf{S0} & \textbf{S1} & \textbf{S2} & \textbf{S3} \\
      \hline
      \textbf{R0} & A & B &  & \\
      \hline
    \end{tabular}
    \hspace{1cm}
    \begin{tabular}{ |l l  l |}
      \hline
      \textbf{Register} & \textbf{Quorum} & \textbf{Decision state} \\
      \hline
      $R0$ & $\{S0,S1,S2\}$ & \textsc{None} \\
      & $\{S0,S1,S3\}$ & \textsc{None} \\
      & $\{S0,S2,S3\}$ & \textsc{Maybe} A \\
      & $\{S1,S2,S3\}$ & \textsc{Maybe} B \\
      \hline
    \end{tabular}

  \caption{Sample client state table (left) and decision table (right) for client $C0$ during Fast Paxos.}
  \label{fig:fpaxos_decisiontwo}
  \end{footnotesize}
\end{figure}

\three~Furthermore, a client in phase one will only need to intersect with any previous two \emph{fast} quorums if it is unable to determine which value to propose.
Figures~\ref{fig:fpaxos_decision} \&~\ref{fig:fpaxos_decisiontwo} give an example of this with the configuration from Figure~\ref{fig:example_configs/4s_intersecting}.
According to Equation~(**), the client $C0$ needs to read three registers from register set 0 before it can safely write to register set 1.
However, in Figure~\ref{fig:fpaxos_decision}, the client can safely write to register set $1$ after reading just two registers.
This is not the case in Figure~\ref{fig:fpaxos_decisiontwo} however.

\begin{figure}
\begin{tcolorbox}
\textbf{Phase 1}
\begin{itemize}[leftmargin=*]
  \item A client $c$ chooses a register set $r$ that is either: quorum intersecting or is client restricted and has been assigned to $c$ but not yet used.
  $c$ sends \textsc{P1a}($r$) to all servers.
  \item Upon receiving \textsc{p1a}($r$), each server checks if register $r$ is unwritten.
  If so, any unwritten registers up to $r-1$ (inclusive) are set to to \emph{nil}.
  The server replies with \textsc{p1b}($r$,$S$) where $S$ is a set of all written registers.
  \item Each time $c$ receives a \textsc{P1a}, it updates its state and decision tables accordingly.
  \uline{If the decision state of all quorums from register sets $0$ to $r-1$ is \textsc{None} or \textsc{Maybe}~$v$} then $c$ chooses $v$ (or if all states are \textsc{None} then its input value) and proceeds to phase two.
  If $c$ times out before completing phase one, it restarts phase one.
\end{itemize}

\textbf{Phase 2}
\begin{itemize}[leftmargin=*]
  \item $c$ sends \textsc{P2a}($r$,$v$) to all servers where $v$ is value chosen at the end of phase one.
  \item Upon receiving \textsc{P2a}($r$,$v$), each server checks if register $r$ is unwritten.
  If so, any unwritten registers up to $r-1$ (inclusive) are set to \emph{nil} and register $r$ is set to $v$.
  The server replies with \textsc{P2b}($r$,$v$).
  \item Each time $c$ receives a \textsc{P2a}, it updates its state and decision tables accordingly.
  \uline{If the decision state of a quorum is \textsc{Decided}~$v$} then $c$ outputs $v$.
  If $c$ times out before completing phase two, it restarts phase one.
\end{itemize}
\end{tcolorbox}
\caption{The Generalised Fast Paxos algorithm.}
\label{fig:gpaxos_alg}
\end{figure}

Figure~\ref{fig:gpaxos_alg} summaries how these generalisation can be combined into a revised Fast Paxos algorithm.
Note that a client can complete a phase once the completion criteria (underlined) has been satisfied even if it has not executed every step.

\section{Example consensus algorithms}
\label{sec:examples}

In this section, we outline three uses of our generalisation of Paxos and Fast Paxos by utilising different configurations.

\textbf{Co-located consensus.}
Consider a configuration which uses a quorum containing all servers for the first $k$ register sets and majority quorums afterwards, as shown in Figure~\ref{fig:add_example_configs/allin}.
All registers sets are client restricted.
Participants in a system may be deciding a value between themselves, and so a server and client are co-located on each participant.
A client can therefore either achieve consensus in one round trip to all servers (if all are available)  or two round trips to any majority (in case a server has failed).

\textbf{Fixed-majority consensus.}
Consider a configuration with one majority quorum for register set $0$ and majority quorums for register sets $1$ onwards, as shown in Figure~\ref{fig:add_example_configs/fixedmaj}.
Register set $0$ is quorum intersecting and register sets $1$ onwards are client restricted.
A client can either achieve consensus in one round trip to a specific majority or two round trips to any majority.

\textbf{Reconfigurable consensus.}
Consider a set of servers partitioned into a primary set and backup set.
Consider a configuration which uses only primary servers for register set $0$ to $k-1$ and only backup servers from register set $k$, as shown in Figure~\ref{fig:add_example_configs/reconfig}.
A client can move the systems from primary servers to backup servers by executing Paxos for register set $k$ or greater.
No subsequent client will need a reply from a primary server to make progress whilst the backup set is available.

\section{Conclusion}
\label{sec:conc}

Paxos has long been the \emph{de facto} approach to reaching consensus, however, this ``one size fits all" solution performs poorly in practice and is famously difficult to understand.
In this paper, we have reframed the problem of distributed consensus in terms of write-once registers and thus proposed a generalised solution to distributed consensus.
We have demonstrated that this solution not only unifies existing algorithms including Paxos and Fast Paxos but also demonstrates that such algorithms are conservative as their quorum intersection requirements and quorum agreement rules can be substantially weakened.
We have illustrated the power of our generalised consensus algorithm by proposing three novel algorithms for consensus, demonstrating a few interesting points on the diverse array of algorithms made possible by our abstract.

Our aim is to make reasoning about correctness sufficiently intuitive that proofs are not necessary to make a convincing case for the safety; nonetheless, we include in Appendix~\ref{sec:correctness_proofs} for completeness.


\section{Acknowledgements}
We would like to thank Jon Crowcroft, Stephen Dolan and Martin Kleppmann for their valuable feedback on this paper.
This work was funded in part by EPSRC EP/N028260/2 and EP/M02315X/1.

\newpage
\bibliographystyle{abbrv}
\bibliography{refs}

\clearpage
\appendixpage
\appendix
\section{Proofs of safety}
\label{sec:correctness_proofs}

In this appendix, we provide proofs for the safety properties (non-triviality, agreement) of our proposed algorithms for solving consensus.

\subsection{Four correctness rules}
\label{subsec:proofs_rules}

Figure~\ref{fig:rules} proposed four rules which we claim are sufficient to satisfy the non-triviality and agreement requirements of distributed consensus (\S\ref{sec:problem}).
We now consider each requirement in turn.
We will use $s[r]=v$ to denote that the value $v$ is in register $r$ on server $s$.

\begin{theorem}[Satisfying non-triviality]
If a value $v$ is the output of a client $c$ then $v$ was the input of some client $c'$.
\end{theorem}

\begin{proof}\ \\
\indent Assume $v$ was the output of client $c$.

According to Rule 1, $\exists r \in \mathbb{N}_{0}, \exists Q \in \mathcal{Q}_r,\forall s \in Q: s[r]=v$ therefore at least one register contains $v$.

Consider the invariant that all (non-nil) registers contain client input values.
Initially, all registers are unwritten thus this invariant holds.
According to Rule 2, each client will only write either their input value or a value copied from another register, thus the invariant will be preserved.
\end{proof}

\begin{theorem}[Satisfying agreement]
If two clients, $c$ and $c'$, output values, $v$ and $v'$ (respectively), then $v=v'$.
\end{theorem}

\begin{proof}\ \\
\indent Assume that value $v$ was the output of client $c$.
Assume that value $v'$ was the output of client $c'$.

According to Rule 1, the following must be true:
\begin{align*}
\exists r \in \mathbb{N}_{0}, \exists Q \in \mathcal{Q}_r,\forall s \in Q: s[r]=v\\
\exists r' \in \mathbb{N}_{0}, \exists Q' \in \mathcal{Q}_r',\forall s' \in Q': s'[r']=v'
\end{align*}

Since register sets are totally ordered, it must be the case that either $r=r'$, $r < r'$ or $r > r'$:
\begin{description}
  \item[Case] $r=r'$:\\ Both decisions are in the same register set.
  It is either the case that both clients have read from the same quorum or they have read from different quorums.
  \begin{description}
    \item[Case] $Q=Q'$:\\ Each quorum can decide at most one value thus $v=v'$
    \item[Case] $Q\neq Q'$:\\ According to Rule 3, since $Q$ has decided $v$, each client who wrote a register in $Q$ must have ensured that no other quorum in register set $r$ can reach a different decisions.
    Thus $v=v'$.
  \end{description}
  \item[Case] $r < r'$:\\ According to Rule 4, a client will only write $v$ to register set $r'$ after ensuring no quorum in register set $r$ will reach a different decision. Thus $v=v'$.
  \item[Case] $r > r'$:\\ This is the same as $r < r'$ with $r$ and $r'$ swapped. Thus $v=v'$.
\end{description}
\end{proof}

\subsection{Client decision table rules}
\label{subsec:proofs_tables}

We have shown that the four rules for correctness are sufficient to satisfy the non-triviality and agreement requirements of consensus.
We will now show that the client decision table rules (Figure~\ref{fig:decisiontable}) implement the four rules for correctness (Figure~\ref{fig:rules}) and thus satisfies the non-triviality and agreement requirements of consensus.

\begin{theorem}[Satisfying Rule 1]
If the value $v$ is the output of client $c$ then $c$ has read $v$ from a quorum $Q \in Q_r$ in register set $r$.
\end{theorem}

\begin{proof}\ \\
\indent Assume the value $v$ is the output of client $c$.
There must exist a register set $r$ and quorum $Q \in \mathcal{Q}_r$ in the decision table of $c$ with the status \textsc{Decided}~$v$ (Figure~\ref{fig:decisiontable}).
A quorum $Q$ can only reach decision state \textsc{Decided}~$v$ if $\forall s \in Q: s[r]=v$.
\end{proof}

\begin{theorem}[Satisfying Rule 2]
If the value $v$ is written by a client $c$ then either $v$ is $c$'s input value or $v$ has been read from some register.
\end{theorem}

\begin{proof}\ \\
\indent Assume the value $v$ has been written by client $c$.
According to Figure~\ref{fig:decisiontable}, $v$ must be either the input value of $c$ or read from some register.
\end{proof}

\begin{theorem}[Satisfying Rule 3]
If the values $v$ and $v'$ are decided in register set $r$ then $v=v'$.
\end{theorem}

\begin{proof}\ \\
\indent Assume the value $v$ is decided in register set $r$ by quorum $Q \in \mathcal{Q}_r$, thus $\forall s \in Q: s[r]=v$.
Assume the value $v'$ is decided in register set $r$ by quorum $Q' \in \mathcal{Q}_r$,  thus $\forall s \in Q': s[r]=v'$.

The register set $r$ either uses intersecting quorums or client restricted configuration.
\begin{description}
  \item[Case] $r$ is client restricted:\\
  Each client is assigned a disjoint subset of register sets thus at most one client is assigned $r$.
  A client will only write a (non-nil) value to $r$ if they have been assigned it and not yet written to it (Figure~\ref{fig:decisiontable}).
  The register set $r$ will therefore only contain one (non-nil) value thus $v=v'$.
  \item[Case] $r$ has intersecting quorums:\\
  This requires that there exists a server $s$ such that $s \in Q$ and $s \in Q'$.
  We require that both $s[r]=v$ and $s[r]=v'$, thus $v=v'$.
\end{description}
\end{proof}

\begin{theorem}[Satisfying Rule 4]
If the value $v$ is decided in register set $r$ and the (non-nil) value $v'$ is written to register set $r'$ where $r<r'$ then $v=v'$
\end{theorem}

We will prove this by induction over the writes to register sets $>r$.

\begin{theorem}[Satisfying Rule 4 - Base case]
If the value $v$ is decided in register set $r$ then the first (non-nil) value to be written to a register set $r'$ where $r<r'$ is $v$.
\end{theorem}

\begin{proof}\ \\
\indent Assume the value $v$ is decided in register set $r$ by quorum $Q \in \mathcal{Q}_r$ thus $\forall s \in Q: s[r]=v$.
Since registers are write once, the following always holds true: $\forall s \in Q: s[r]=v \lor s[r]=unwritten$.

Assume the value $v'$ is written to register set $r'$ by client $c$ where $r<r'$.
Assume that $v'$ is the first value to be written thus $c$ cannot read any (non-nil) values from registers $>r$ before writing $v'$.

We will show that $v=v'$.

Consider the decision table of client $c$ when it is writing to $r'$.
Since $r<r'$, the decision state of $Q$ must be either \textsc{None}, \textsc{Maybe}~$v'$ or \textsc{Decided}~$v'$ (Figure~\ref{fig:decisiontable}).
\begin{description}
  \item[Case] \textsc{Decided} $v'$:\\
  This decision state requires that $\forall s \in Q: s[r]=v'$.
  Since we know that $\forall s \in Q: s[r]=v \lor unwritten$  then $v=v'$.
  \item[Case] \textsc{Maybe} $v'$:\\
  The decision state \textsc{Maybe} $v'$ can be reached in one of three ways:
  \begin{description}
    \item[Case] $c$ read $v'$ from register $r$ of some server $s \in Q$:
    Since we know that $\forall s \in Q: s[r]=v \lor unwritten$ then $v=v'$.
    \item[Case] $r$ is client restricted and $c$ read $v'$ from register $r$ of some server $s$
    Since we know that $\forall s \in Q: s[r]=v \lor unwritten$  then $v=v'$.
    \item[Case] $c$ read $v'$ from a register $>r$:
    Since $v'$ is the first value to be written to a register $>r$, this case cannot occur.
  \end{description}
  \item[Case] \textsc{None}:\\
  The decision state \textsc{None} can be reached in one of five ways:
  \begin{description}
    \item[Case] $c$ read \emph{nil} from register $r$ of some server $s \in Q$:\\
    Since $\forall s \in Q: s[r]=v \lor unwritten$, this case cannot occur.
    \item[Case] $c$ read two different values from two servers, $s,s' \in Q$:\\
    Since $\forall s \in Q: s[r]=v \lor unwritten$, this case cannot occur.
    \item[Case] $c$ read two different values from registers $>r$:\\
    Since $v'$ is the first value to be written to a register $>r$, this case cannot occur.
    \item[Case] $c$ read a value from register $r$ of some server $s \in Q$ and a different value from a register $>r$:\\
    Since $v'$ is the first value to be written to a register $>r$, this case cannot occur.
    \item[Case] $r$ is client restricted, $c$ read a value from a register in $r$ and a different value from a register $>r$:\\
    Since $v'$ is the first value to be written to a register $>r$, this case cannot occur.
  \end{description}
\end{description}
\end{proof}

Since the following proof overlaps significantly with the previous proof, we have underlined the parts which have been altered.

\begin{theorem}[Satisfying Rule 4 - Inductive case]
If the value $v$ is decided in register set $r$ and all (non-nil) values written to registers $>r$ are $v$ then the next (non-nil) value to be written to a register set $r'$ where $r<r'$ is also $v$.
\end{theorem}

\begin{proof}\ \\
\indent
Assume the value $v$ is decided in register set $r$ by quorum $Q \in \mathcal{Q}_r$ thus $\forall s \in Q: s[r]=v$.
Since registers are write once, the following always holds true: $\forall s \in Q: s[r]=v \lor s[r]=unwritten$.

Assume the value $v'$ is written to register set $r'$ by client $c$ where $r<r'$.
\uline{Assume that all (non-nil) values written to registers $>r$ are $v$ thus $c$ can only read $v$ from (non-nil) registers $>r$.}

We will show that $v=v'$.

Consider the decision table of client $c$ when it is writing to $r'$.
Since $r<r'$, the decision state of $Q$ must be either \textsc{None}, \textsc{Maybe}~$v'$ or \textsc{Decided}~$v'$ (Figure~\ref{fig:decisiontable}).
\begin{description}
  \item[Case] \textsc{Decided} $v'$:\\
  This decision state requires that $\forall s \in Q: s[r]=v'$.
  Since we know that $\forall s \in Q: s[r]=v \lor unwritten$  then $v=v'$.
  \item[Case] \textsc{Maybe} $v'$:\\
  The decision state \textsc{Maybe} $v'$ can be reached in one of three ways:
  \begin{description}
    \item[Case] $c$ read $v'$ from register $r$ of some server $s \in Q$:\\
    Since we know that $\forall s \in Q: s[r]=v \lor unwritten$ then $v=v'$.
    \item[Case] $r$ is client restricted and $c$ read $v'$ from register $r$ of some server $s$:\\
    Since we know that $\forall s \in Q: s[r]=v \lor unwritten$  then $v=v'$.
    \item[Case] $c$ read $v'$ from a register $>r$:\\
    \uline{Since $v$ is the only (non-nil) value to be written to registers $>r$ then $v=v'$.}
  \end{description}
  \item[Case] \textsc{None}:\\
  The decision state \textsc{None} can be reached in one of five ways:
  \begin{description}
    \item[Case] $c$ read \emph{nil} from register $r$ of some server $s \in Q$:\\
    Since $\forall s \in Q: s[r]=v \lor unwritten$, this case cannot occur.
    \item[Case] $c$ read two different values from two servers, $s,s' \in Q$:\\
    Since $\forall s \in Q: s[r]=v \lor unwritten$, this case cannot occur.
    \item[Case] $c$ read two different values from registers $>r$:\\
    \uline{Since $v$ is the only (non-nil) value to be written to registers $>r$, this case cannot occur.}
    \item[Case] $c$ read a value from register $r$ of some server $s \in Q$ and a different value from a register $>r$:\\
    \uline{Since we know that $\forall s \in Q: s[r]=v \lor unwritten$ and $v$ is the only (non-nil) value to be written to registers $>r$, this case cannot occur.}
    \item[Case] $r$ is client restricted, $c$ read a value from a register in $r$ and a different value from a register $>r$:\\
    \uline{Since we know that at some time $\forall s \in Q: s[r]=v$, then if $r$ is client restricted then all non-nil registers in $r$ must contain $v$.
    Since $v$ is the only (non-nil) value to be written to registers $>r$, this case cannot occur. }
  \end{description}
\end{description}
\end{proof}

\subsection{(Fast) Paxos}
\label{subsec:proofs_paxos}

Figure~\ref{fig:paxos_alg} describes the Paxos algorithm using write-once registers.
Section~\ref{sec:fast_paxos} describe how to generalise Figure~\ref{fig:paxos_alg} to Fast Paxos.
In this section, we proof that Fast Paxos (and therefore Paxos) implements the four rules for correctness (Figure~\ref{fig:rules}) and thus satisfies the non-triviality and agreement requirements of consensus.

\begin{theorem}[Satisfying Rule 1]
If the value $v$ is the output of client $c$ then $c$ has read $v$ from a quorum $Q \in Q_r$ in register set $r$.
\end{theorem}

\begin{proof}\ \\
\indent Assume the value $v$ is the output of client $c$.

This must be the result of $c$ completing phase two of Fast Paxos for some register set $r$.
$c$ must have received the message \textsc{P2b}($r$,$v$) from $>\frac{1}{2}$/$\geq\frac{3}{4}$ of servers (depending on either $r$ is \emph{classic}/\emph{fast}).
Prior to sending \textsc{P2b}($r$,$v$), each server $s$ has written register $r$ to $v$.
$Q_r$ is any subset of servers containing $>\frac{1}{2}$/$\geq\frac{3}{4}$ of servers (depending on either $r$ is \emph{classic}/\emph{fast}).
Thus $c$ has read a quorum $Q \in Q_r$ in register set $r$.
\end{proof}

\begin{theorem}[Satisfying Rule 2]
If the value $v$ is written by a client $c$ then either $v$ is $c$'s input value or $v$ has been read from some register.
\end{theorem}

\begin{proof}\ \\
\indent Assume a value $v$ is written by a client $c$.
This must be the result of completing phase one of Fast Paxos for some register set $r$ and choosing the value $v$.
The value $v$ must have been chosen in one of following ways:
\begin{description}
  \item[Case] $0$ (non-nil) registers where returned with \textsc{P1b} messages:\\
  In this case, $v$ is $c$'s input value.
  \item[Case] $1$ or more (non-nil) registers where returned with \textsc{P1b} messages:\\
  In this case, $v$ is the most common value from the greatest register set thus $v$ has been read from some register.
\end{description}
\end{proof}

\begin{theorem}[Satisfying Rule 3]
If the values $v$ and $v'$ are decided in register set $r$ then $v=v'$.
\end{theorem}

\begin{proof}\ \\
\indent Assume the values $v$ and $v'$ are decided in register set $r$.
It is therefore the case that there exists two quorums $Q,Q' \in Q_r$ such that $\forall s \in Q: s[r]=v$ and $\forall s \in Q': s[r]=v'$
The register set $r$ is either \emph{fast} (quorum intersecting) or \emph{classic} (client restricted):
\begin{description}
  \item[Case] $r$ is \emph{fast}:\\
  There exists a server $s$ where $s \in Q$ and $s \in Q'$.
  We require that $s[r]=v \land s[r]=v'$ thus $v=v'$
  \item[Case] $r$ is \emph{classic}:\\
  At most one client is assigned register set $r$.
  Each client only writes (non-nil) values to assigned register sets and each does so with only one value.
  Therefore $v=v'$.
\end{description}
\end{proof}

\begin{theorem}[Satisfying Rule 4]
If the value $v$ is decided in register set $r$ and the (non-nil) value $v'$ is written to register set $r'$ where $r<r'$ then $v=v'$
\end{theorem}

We will prove this by induction over the writes to register sets $>r$.

\begin{theorem}[Satisfying Rule 4 - Base case]
If the value $v$ is decided in register set $r$ then the first (non-nil) value to be written to a register set $r'$ where $r<r'$ is $v$.
\end{theorem}

\begin{proof}\ \\
\indent Assume the value $v$ is decided in register set $r$.
If $r$ is \emph{fast} (quorum intersecting), $v$ must have been written to register $r$ on $\frac{3}{4}$ or more of servers.
Otherwise, if $r$ is \emph{classic} (client restricted), $v$ must have been written to register $r$ on least $\frac{1}{2}$ of servers.
The writing of $v$ to $r$ must be the result of receiving \textsc{P2a}($r$,$v$).

Assume the (non-nil) value $v'$ is written to register set $r'$ by client $c$.
This must be the result of completing phase one of Fast Paxos for register set $r'$ and choosing the value $v'$.
The value $v'$ could be chosen in one of two ways:
\begin{description}
  \item[Case] $v'$ is $c$'s input value:
  This requires that (non-nil) registers where returned to $c$ with the \textsc{P1b} messages for $r$.
  At last one server $s$ must both write $s[r]=v$ and send a \textsc{P1b} message to $c$ since both require at least $\frac{1}{2}$ of servers.
  \begin{description}
    \item[Case] $s$ sends \textsc{P1b} for register $r'$ first:\\
    Prior to sending \textsc{P1b}, the server $s$ must write \emph{nil} to all unwritten registers $0$ to $r'-1$, including register $r$ since $r<r'$.
    Server $s$ will not be able to later write $s[r]=v$ so this case cannot occur.
    \item[Case] $s$ must write $s[r]=v$ first:\\
    Since no registers where returned with \textsc{P1b} messages, this case cannot occur.
  \end{description}
  \item[Case] $v'$ is the most common value read from the greatest (non-nil) register set:
  This requires that $1$ or more (non-nil) registers where returned to $c$ with the \textsc{P1b} messages for $r$.
  As we have already seen, at least one \textsc{P1b} message for register $r'$ must include $r:v$.
  The chosen value $v'$ must have either been read from register set $r$ or from any register set $>r$.
  \begin{description}
      \item[Case] $v'$ was read from register set $r$:\\
      The register set $r$ is either \emph{fast} (quorum intersecting) or \emph{classic} (client restricted):
      \begin{description}
          \item[Case] $r$ is \emph{classic}:\\
          All (non-nil) registers from $r$ returned with \textsc{P1b} messages will contain $v$ thus $v=v'$.
          \item[Case] $r$ is \emph{fast}:\\
          At least $\frac{1}{4}$ of servers will reply with $r:v$.
          Therefore $v$ will be the most common value and it will be chosen by $c$ thus $v=v'$.
      \end{description}
      \item[Case] $v'$ was read from a register set $>r$:\\
      Since the client $c$ is the first to write to a register $>r$ then $c$ will not read any registers $>r$. Therefore this case cannot occur.
      \end{description}
\end{description}
\end{proof}

Since the following proof overlaps significantly with the previous proof, we have underlined the parts which have been altered.

\begin{theorem}[Satisfying Rule 4 - Inductive case]
If the value $v$ is decided in register set $r$ and all (non-nil) values written to registers $>r$ are $v$ then the next (non-nil) value to be written to a register set $r'$ where $r<r'$ is also $v$.
\end{theorem}

\begin{proof}\ \\
\indent Assume the value $v$ is decided in register set $r$.
If $r$ is \emph{fast} (quorum intersecting), $v$ must have been written to register $r$ on $\frac{3}{4}$ or more of servers.
Otherwise, if $r$ is \emph{classic} (client restricted), $v$ must have been written to register $r$ on least $\frac{1}{2}$ of servers.
The writing of $v$ to $r$ must be the result of receiving \textsc{P2a}($r$,$v$).

Assume the (non-nil) value $v'$ is written to register set $r'$ by client $c$.
This must be the result of completing phase one of Fast Paxos for register set $r'$ and choosing the value $v'$.
The value $v'$ could be chosen in one of two ways:
\begin{description}
  \item[Case] $v'$ is $c$'s input value:
  This requires that (non-nil) registers where returned to $c$ with the \textsc{P1b} messages for $r$.
  At last one server $s$ must both write $s[r]=v$ and send a \textsc{P1b} message to $c$ since both require at least $\frac{1}{2}$ of servers.
  \begin{description}
    \item[Case] $s$ sends \textsc{P1b} for register $r'$ first:\\
    Prior to sending \textsc{P1b}, the server $s$ must write \emph{nil} to all unwritten registers $0$ to $r'-1$, including register $r$ since $r<r'$.
    Server $s$ will not be able to later write $s[r]=v$ so this case cannot occur.
    \item[Case] $s$ must write $s[r]=v$ first:\\
    Since no registers where returned with \textsc{P1b} messages, this case cannot occur.
  \end{description}
  \item[Case] $v'$ is the most common value read from the greatest (non-nil) register set:
  This requires that $1$ or more (non-nil) registers where returned to $c$ with the \textsc{P1b} messages for $r$.
  As we have already seen, at least one \textsc{P1b} message for register $r'$ must include $r:v$.
  The chosen value $v'$ must have either been read from register set $r$ or from any register set $>r$.
  \begin{description}
      \item[Case] $v'$ was read from register set $r$:\\
      The register set $r$ is either \emph{fast} (quorum intersecting) or \emph{classic} (client restricted):
      \begin{description}
          \item[Case] $r$ is \emph{classic}:\\
          All (non-nil) registers from $r$ returned with \textsc{P1b} messages will contain $v$ thus $v=v'$.
          \item[Case] $r$ is \emph{fast}:\\
          At least $\frac{1}{4}$ of servers will reply with $r:v$.
          Therefore $v$ will be the most common value and it will be chosen by $c$ thus $v=v'$.
      \end{description}
      \item[Case] $v'$ was read from a register set $>r$:\\
      \uline{Since all non-nil registers $>r$ contain $v$ then $c$ will not read any other value from any registers $>r$ thus $v=v'$.}
  \end{description}
\end{description}
\end{proof}

\end{document}